\documentclass[conference]{IEEEtran}
\IEEEoverridecommandlockouts
\usepackage{cite}
\usepackage{amsmath,amssymb,amsfonts, amsthm}
\usepackage{algorithmic}
\usepackage{graphicx}
\usepackage{textcomp}
\usepackage{xcolor}
\def\BibTeX{{\rm B\kern-.05em{\sc i\kern-.025em b}\kern-.08em
    T\kern-.1667em\lower.7ex\hbox{E}\kern-.125emX}}

\usepackage{multirow}
\usepackage{comment}

\def\bea{\begin{eqnarray}}
\def\eea{\end{eqnarray}}

\newtheorem{theorem}{Theorem}

\newtheorem{prop}{Proposition}

\begin{document}

\title{Feature Map for Quantum Data in Classification}

\author{
\IEEEauthorblockN{1\textsuperscript{st} Hyeokjea Kwon}
\IEEEauthorblockA{\textit{School of Electrical Engineering} \\
\textit{Korea Adv. Inst. of Sci. and Tech.}\\
Daejeon, Yuseong-gu \\
hyukjae.kwon94@gmail.com}
\and
\IEEEauthorblockN{2\textsuperscript{nd} Hojun Lee}
\IEEEauthorblockA{\textit{School of Electrical Engineering} \\
\textit{Korea Adv. Inst. of Sci. and Tech.}\\
Daejeon, Yuseong-gu \\
quantum0430@kaist.ac.kr}
\and
\IEEEauthorblockN{3\textsuperscript{rd} Joonwoo Bae}
\IEEEauthorblockA{\textit{School of Electrical Engineering} \\
\textit{Korea Adv. Inst. of Sci. and Tech.}\\
Daejeon, Yuseong-gu \\
bae.joonwoo@gmail.com}
}

\maketitle

\begin{abstract}
The kernel trick in supervised learning signifies transformations of an inner product by a feature map, which then restructures training data in a larger Hilbert space according to an endowed inner product. A quantum feature map corresponds to an instance with a Hilbert space of quantum states by fueling quantum resources to machine learning algorithms. In this work, we point out that the quantum state space is specific such that a measurement postulate characterizes an inner product and that manipulation of quantum states prepared from classical data cannot enhance the distinguishability of data points. We present a feature map for quantum data as a probabilistic manipulation of quantum states to improve supervised learning algorithms. 
\end{abstract}


\section{Introduction}
In supervised learning, one aims to construct a model that makes predictions based on training data. Recently, the framework has begun to apply the laws of quantum mechanics and quantum machine learning, to fuel nonclassical properties such as entanglement and superposition to machine learning (ML) algorithms for further advantages, see e.g., \cite{article_qml, book_slqc, PhysRevLett.117.130501}. One way to apply Quantum Information Theory (QIT) in ML is to process ML algorithms with quantum states prepared according to classical data. 

From the view of QIT, the state preparation rephrases embedding classical data to quantum systems. Technically, the space of quantum states is described by a Hilbert space where the measurement postulate, called the Born rule, specifies an inner product \cite{book_mfqm}. From the view of ML, embedding data in a Hilbert space corresponds to a feature map: its quantum application is called a {\it quantum feature map} \cite{PhysRevLett.122.040504, schuld2021supervised}. Then, the resulting quantum states and the Hilbert space are referred to as feature vectors and a feature space, respectively. 

On the one hand, a quantum feature map allows one to exploit quantum resources such as entanglement and superposition existing in quantum states to enhance ML algorithms. As a result, one may envisage quantum advantages over classical counterparts. On the other hand, one notices that the quantum state space is a specific and restricted object. It is a Hilbert space entirely characterized by the postulates of quantum theory \cite{book_mfqm}. 

The consequences show that once a feature map prepares quantum states, quantum operations are contractive, i.e., the norm of feature vectors does not increase \cite{article_cptpm}. Moreover, a mathematical space describing quantum states is not hypothetical: the measurement postulate defines an inner product uniquely in the space, known as the Gleason theorem \cite{49efcab3-5adb-3b03-b361-7c2f8afacf95, PhysRevLett.91.120403}. The uniqueness implies limitations on the so-called {\it kernel tricks} in the quantum feature space. Apart from the fact the quantum state preparation corresponds to a feature map {\it per se}, little is known about how quantum principles can be incorporated into feature vectors to enhance ML algorithms.

In this work, we show that once quantum states are prepared for an ML algorithm, their distinguishability does not increase by a feature map. In contrast to classical ML algorithms, quantum data cannot be manipulated such that their distinguishability is enhanced. We then present a general manipulation of quantum data, namely a {\it feature map for quantum data}, by relaxing the trace-preserving condition, as a versatile tool to improve ML algorithms. We also develop a quantum circuit construction of a feature map for quantum data and demonstrate its advantages in supervised learning for binary classification in different datasets.

\section{Quantum classification}

Let us write by $\mathcal{D}$ a dataset,
\begin{equation}
\mathcal{D}=\{ (x_1, y_1), \cdots, (x_M, y_M) \},\nonumber
\end{equation}
called a training dataset, where $x_m\in\mathcal{X}$ are data and $y_m\in\mathcal{Y}$ are their labels for classification. A supervised ML algorithm constructs a classifier $f\in\mathcal{F}$ from a provided dataset, aiming to make a classification with a higher precision for an item $x_m$ apart from a dataset is given. The statistical learning theory then introduces an optimization problem,
\begin{equation}
\underset{f\in\mathcal{F}}{\arg \min}~ \hat{\mathcal{R}}_{L, \mathcal{D}}(f) + g(\| f \|_{\mathcal{F}}), \label{eq:SupLearning}
\end{equation}
where $\hat{\mathcal{R}}_{L, \mathcal{D}}(f)$ is an empirical risk of a function $f$, and $g$ is a regularization function that takes a list of constraints, which should be fulfilled, into account \cite{hastie01statisticallearning}. The empirical risk can be expressed by a loss function, denoted by $L$, which relies on a dataset $\{x_m \}$ and classification $\{ y_m\}$ as follows,
\begin{equation}
\hat{R}_{L, \mathcal{D}}(f) = \frac{1}{M}\sum_{m=1}^{M}L(f(x_m), y_m). \nonumber
\end{equation}
An instance of a loss function may be given by a distance, e.g., $L(f(x_m), y_m) = |f(x_m)-y_m|^2$. Thus, an optimization problem of an empirical risk can be equivalently formulated by minimizing an average loss function.

A useful technique in the optimization problem \eqref{eq:SupLearning} can be facilitated, which is to devise a mapping from the data $\mathcal{X}$ to a feature space $\mathcal{F}$ as follows,
\begin{equation}
\phi:\mathcal{X} \to \mathcal{F}. \nonumber
\end{equation}
Then, $x_m\in\mathcal{X}$ results in a feature vector $\phi(x_m) \in \mathcal{F}$. An inner product of feature vectors defines a kernel, 
\begin{equation}
\kappa : \mathcal{X} \times \mathcal{X} \to \mathbb{R}, ~~\mathrm{where} ~~
\kappa(x_n, x_m) = \langle \phi(x_n), \phi(x_m) \rangle_{\mathcal{F}}. \nonumber
\end{equation}
The Moore-Aronszajn theorem tells the uniqueness of a reproducing kernel Hilbert space (RKHS) for a given kernel. Therefore, one can rephrase that a feature map $\phi$ introduces a unique RKHS. 

\begin{theorem}[Representer Theorem \cite{article_grt}]
Let $\mathcal{D}$ be a dataset consisting of pairs $(x_m, y_m) \in \mathcal{X} \times \mathcal{Y}$, $\kappa: \mathcal{X} \times \mathcal{X} \to \mathbb{R}$ be a kernel, and $f: \mathcal{X} \to \mathbb{R}$ be a class of model functions in the RKHS. Then the optimal solution of \eqref{eq:SupLearning} can be expressed as
\begin{equation}
f(x) = \sum_{m=1}^{M} \alpha_m \kappa(x_m, x), \label{eq:RepTheorem}
\end{equation}
where $\alpha_m\in\mathbb{R}$ are real parameters.
\end{theorem}

\subsection*{Quantum binary classification}

\begin{figure}[t]
	\begin{center}
		\includegraphics[angle=0, width=.48 \textwidth]{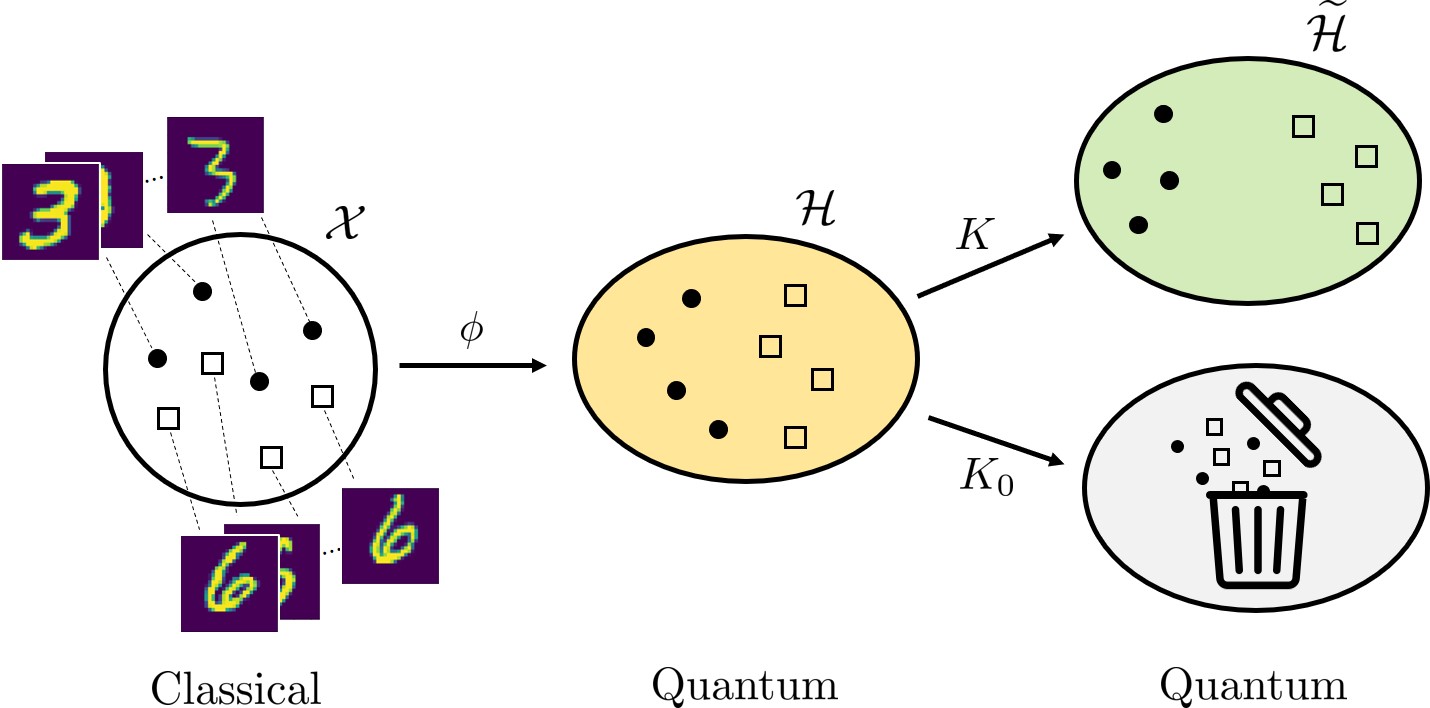}
		\caption{ A quantum embedding $\phi$ corresponds to the preparation quantum states according to classical data, after which distances among quantum data can be defined by the inner product on the state space, a Hilbert space. We introduce a feature map, a probabilistic transformation characterized a Kraus operator $K$ on quantum states; unsuccessful instances by the other one $K_0 = \sqrt{I- K^{\dagger} K}$ are discarded. Once the feature map is successful, distances among quantum data may increase and thus states are better distinguishable.  } \label{fig:ProbManipulation}
	\end{center}	
\end{figure}
 
Let a quantum circuit $U_{\phi}$ prepares a quantum state $\rho_{\phi}(x_m)$ with a data $x_m$. Then, we can define a classical-to-quantum feature map that encodes the classical data to the quantum state. This process is called quantum encoding or quantum embedding.
\begin{equation}
\phi : x_m\in\mathcal{X} \to \rho_{\phi}(x_m)\in S(\mathcal{H}),
\end{equation}
where $\rho_{\phi}(x_m)$ is a quantum state on a Hilbert space $\mathcal{H}$ and $S(\mathcal{H})$ denotes the set of quantum state. This feature map defines a kernel on the Hilbert space, where the inner product is restricted by the measurement postulate.
\begin{equation}
\kappa(x_n, x_m) = \mathrm{tr}[\rho_\phi(x_n)\rho_\phi(x_m)]
\end{equation}
Therefore, the Hilbert space is the RKHS that can deal with supervised ML by the representer theorem.

Consider the binary classification problem where it has the data $\mathcal{X}=\mathbb{R}^{N}$ and the label set $\mathcal{Y} = \{1, -1\}$. In this scenario, we will only consider the empirical risk in \eqref{eq:SupLearning}. Denote the two classes A and B in the data $\mathcal{X}$ that have labels $y_m=1$ and $y_m=-1$, respectively. Then, we can define the ensemble of two quantum states of each class
\begin{equation}
\rho = \frac{1}{M_A}\sum_{y_m=1}\rho_{\phi}(x_m), ~~~ \sigma = \frac{1}{M_B}\sum_{y_m=-1}\rho_{\phi}(x_m), \label{eq:OriState}
\end{equation}
where $M_A = |\{ x_m:y_m=1 \}|$ and $M_B = |\{ x_m:y_m=-1 \}|$ that are the number of the data in each class, respectively. From \eqref{eq:RepTheorem}, we obtain the fidelity classifier by setting parameters $\alpha_m = 1/M_A$ and $\alpha_m = -1/M_B$ when $y_m=1$ and $y_m=-1$, respectively \cite{lloyd2020quantum}.
\begin{equation}
f_{\mathrm{fid}}(x) = \mathrm{tr}[(\rho-\sigma)\rho_{\phi}(x)] \label{eq:FidClassifier}
\end{equation}
The fidelity classifier quantifies how close a test state $\rho_{\phi}(x)$ is to an ensemble $\rho$ or $\sigma$. For instance, for $x$ such that $f_{\mathrm{fid}} (x) >0$, we conclude a state $\rho_{\phi}(x)$ closer to an ensemble $\rho$.

Consider the empirical risk with a \textit{weighted linear function},
\begin{equation}
\begin{aligned}
L_w(f(x_m), y_m) = -w_m f(x_m) y_m, ~\mathrm{where} \\
w_m=
\begin{cases} 
M/M_A ~ \mathrm{if} ~ y_m=1 \\
M/M_B ~ \mathrm{if} ~ y_m=-1. \label{eq:9}
\end{cases}
\end{aligned}
\end{equation}
Note that weights $\{ w_m \}$ can be used to adjust loss differences such that the empirical risk of the fidelity classifier would be given as follows,
\begin{equation}
\hat{R}_{L_w, \mathcal{D}}(f_{\mathrm{fid}}) = -D_{\mathrm{hs}}(\rho, \sigma)  \nonumber
\end{equation}
where the Hilbert-Schmidt distance $D_{\mathrm{hs}}$ is defined by $\mathrm{tr}[(\rho-\sigma)^2]$. Hence, a quantum embedding that maximizes the Hilbert-Schmidt distance of two ensembles $\rho$ and $\sigma$ builds a fidelity classifier with the least empirical risk \cite{lloyd2020quantum}. 

\begin{figure*}[t]
	\begin{center}
		\includegraphics[angle=0, width=.86 \textwidth]{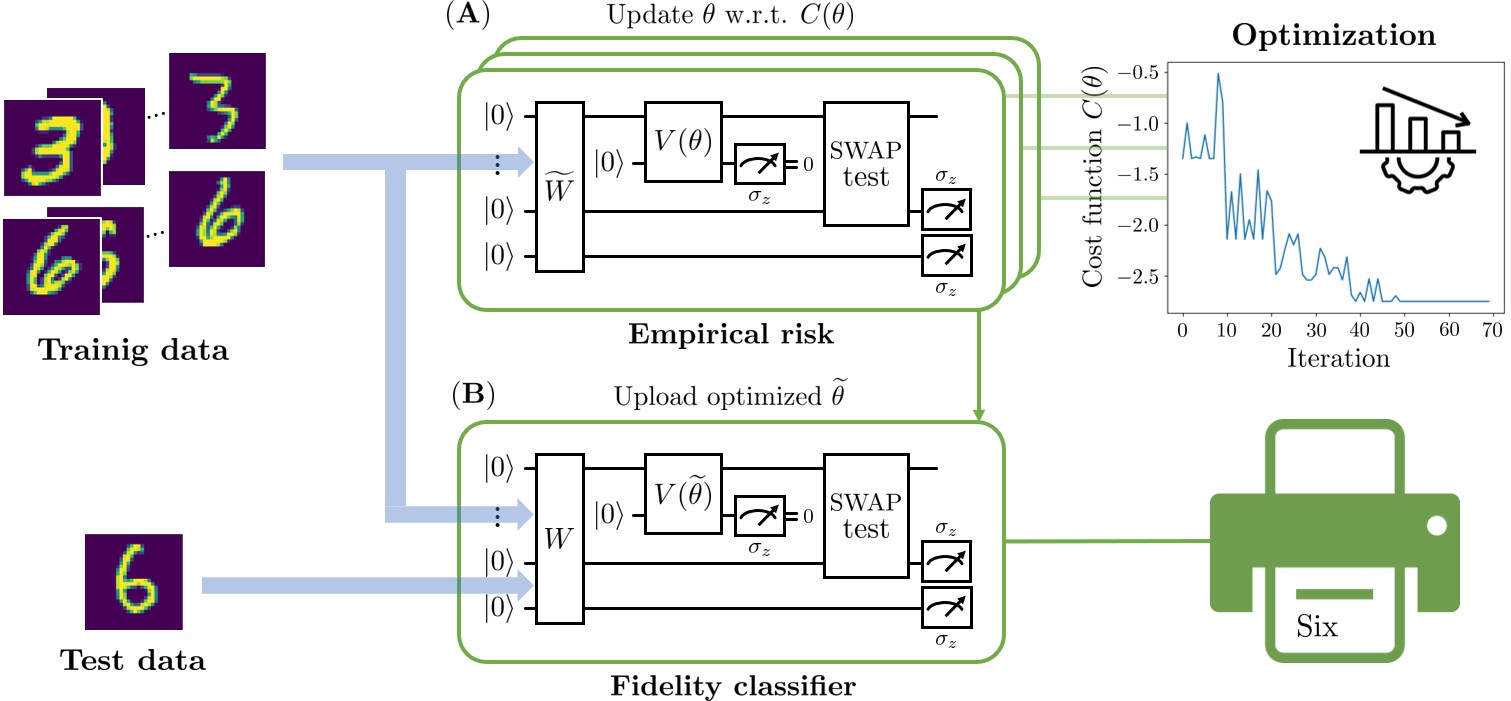}
		\caption{The schematic of quantum binary classification with a proposed feature map is shown; the circuit structure is detailed in Fig. \ref{fig:QuaCircuit}. (A) A unitary transform $\widetilde{W}$ implements quantum embedding according to training data. Then, the quantum data may be transformed by a feature map, realized by a gate $V(\theta)$ over a system and ancillary qubits such that a transform on system qubits is accepted only when an ancillary one gives an outcome $0$. An empirical risk is computed via a SWAP test. Note that a gate $V(\theta)$ is trained to minimize a cost function $C(\theta)$ in (\ref{eq:cost}). (B) A fidelity classifier is realized as follows. A unitary transform $W$ implements state preparation for training and test data. A feature map with $V(\widetilde{\theta})$ over the system and ancillary qubits is applied. Then, a SWAP test is performed for the states resulting from a feature map.} \label{fig:OverScheQuanBinaClassification}
	\end{center}	
\end{figure*}

However, the feature space of a quantum system has restricted the inner product by a measurement postulate. This implies the limitations of a quantum embedding that contrasts with the original feature map, the quantum data cannot be repeatedly embedded in some other feature space such that they are structured with a higher distinguishability.

\begin{prop}
Quantum data cannot be embedded in a feature space with enhanced distinguishability
\end{prop}
\begin{proof}
We first consider cases that map quantum data directly back to a classical feature space by measurements. It is clear that non-orthogonal states cannot be perfectly distinguished, and thus the mapping introduces either an error \cite{Helstrom:1969fri} or ambiguous outcomes \cite{IVANOVIC1987257, DIEKS1988303, PERES198819}, see also \cite{Bergou_2007, Bae_2015}. There are fundamental limitations in the bounds \cite{PhysRevA.77.012113, PhysRevLett.107.170403, 10.1063/1.3298647}

Or, one can consider {\it a feature map for quantum data} as embedding quantum data to high-dimensional quantum systems and then a measurement. A quantum channel, denoted by $\Lambda : \mathcal{H} \rightarrow \widetilde{\mathcal{H}}$ for $\mathrm{dim}(\mathcal{H}) \leq \mathrm{dim}(\widetilde{\mathcal{H}})$, corresponds to a positive and completely positive map for quantum states. Such a map does not increase pairwise distinguishability: for two states $\rho$ and $\sigma$ appearing with probabilities $p_1$ and $p_2$, it holds that
\begin{equation}
\| \Lambda[X] \|_1 \leq \| X \|_1,~\mathrm{for}~X= p_1 \rho - p_2 \sigma \ngeq 0
\end{equation}
where $\|\cdot \|_1$ denotes the $L_1$ norm, i.e., $\|A \|_1 = \mathrm{tr}\sqrt{A^{\dagger}A}$.
\end{proof}

\section{Probabilistic manipulation}
\begin{figure*}[t]
	\begin{center}
		\includegraphics[angle=0, width=1. \textwidth]{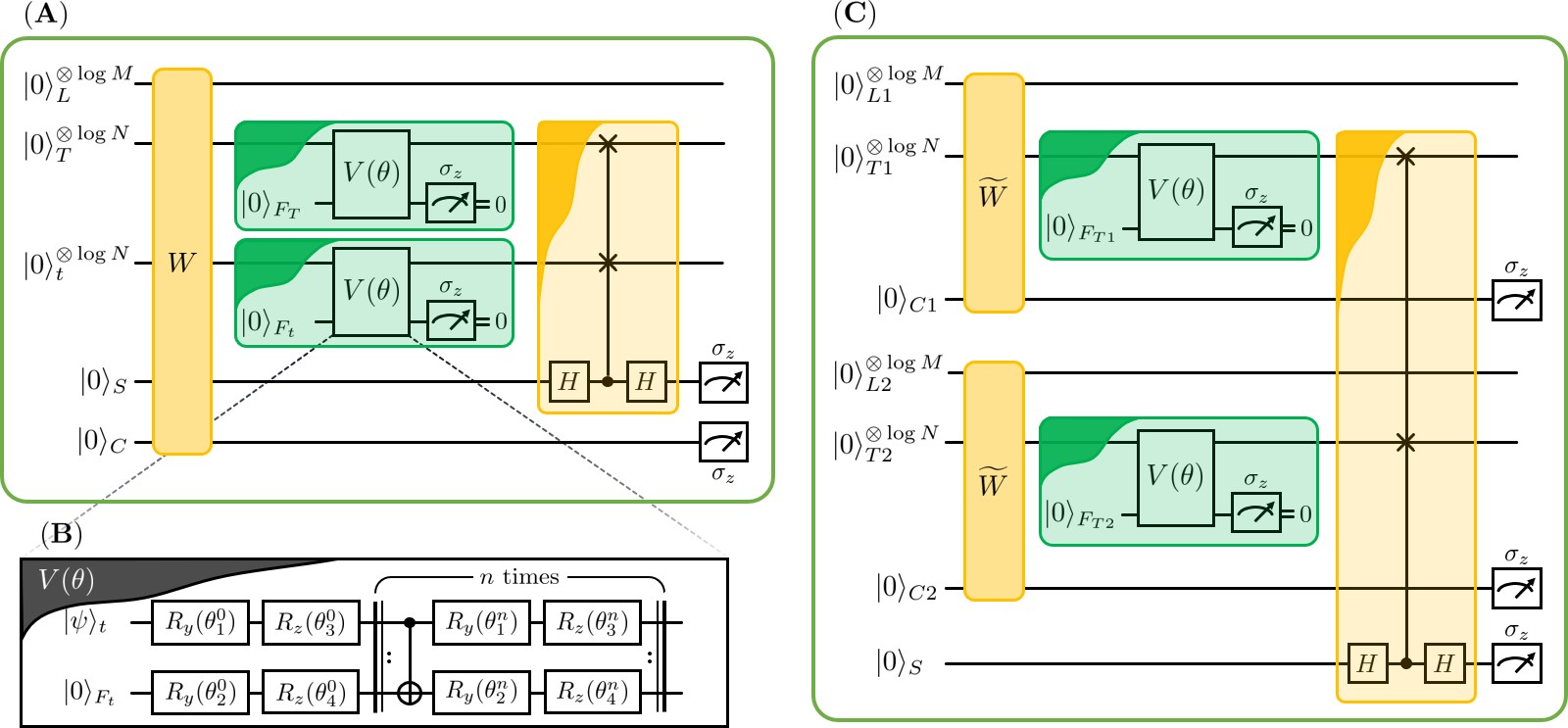}
		\caption{ (A) A quantum circuit for a fidelity classifier in \eqref{eq:NewFidClassifier}  with a proposed feature map is shown. Let $T$ denote the register with training data and $t$ the one with test data. A unitary transformation $W$ implements a state preparation in \eqref{eq:InitState}. A feature map for quantum data corresponds to a probabilistic transformation for each training and test data; the map is realized only when measurement outcomes in both registers $T$ and $t$ are $0$. A parameter $\theta$ in an interaction $V(\theta)$ is trained to optimize a feature map. Then, a Hadamard classifier is implemented with a SWAP test. (B) An interaction $V(\theta)$ for realizing a feature map is trained by parameterized quantum circuits, single- and two-qubit gates. (C) A quantum circuit for an empirical risk with a proposed feature map shares a similar structure to the circuit in (A). It then prepares a state for the training data by a gate $\widetilde{W}$ in \eqref{eq:EmpInitState}. } \label{fig:QuaCircuit}
	\end{center}	
\end{figure*}

We put a step forward to relaxing quantum channels to quantum filtering operations and present a general form of {\it a feature map for quantum data} to enhance ML algorithms with quantum states. To this end, we recall that a quantum channel can be described as the dynamics of a subsystem that interacts with an ancilla through a unitary transformation. When a state and an ancilla initialized in $\rho_{S}\otimes |a\rangle_A\langle a|$ result in $U_{SA} (\rho_S \otimes |a \rangle_A\langle a |) U_{SA}^{\dagger}$ for some interaction $U_{SA}$, a measurement on an ancilla in an orthonormal basis $\{|i\rangle_A \}$ finds the probability of having an outcome $i$,
\begin{equation}
p_A (i) = \mathrm{tr}[U_{SA} \rho_S \otimes |a \rangle_A\langle a | U_{SA}^{\dagger} |i \rangle_A\langle i |]. 
\end{equation}
Then, the resulting state of a system can be described by
\begin{equation}
\rho_S  (i) =\frac{1}{p_A (i)} {K_i} \rho K_{i}^{\dagger}~~\mathrm{where} ~K_i= _A\langle i | U_{SA} |a \rangle_A \label{eq:kraus}
\end{equation}
where $\{ K_i\}$ are called Kraus operators, satisfying the relation $\sum_i K_{i}^{\dagger}K_i  =\mathbb{I}$. 

It is worth mentioning that probabilistic manipulations of quantum states by exploiting Kraus operators have been used to resolve non-trivial problems in various contexts of quantum information applications. In entanglement theory, two-qubit entangled states can be transformed to a more entangled one with some probability by local operations and classical operations, called local filtering \cite{PhysRevA.64.010101, PhysRevLett.89.170401}. The protocol for distilling entanglement can be rephrased as a sequence of local filtering operations \cite{PhysRevLett.76.722, PhysRevLett.77.2818}. Local filtering operations can also reveal hidden nonlocality existing in some entangled states \cite{PhysRevLett.74.2619}. In experiments, a post-selection technique can be described by Kraus operators. For instance, one way to demonstrate a quantum gate with photonic qubits which hardly interact with each other is to select particular measurement outcomes whenever photon-photon interactions were successful, see e.g., \cite{PhysRevLett.107.160401}.

We now present {\it a feature map for quantum data} via Kraus operators as a probabilistic strategy of manipulating quantum states to enhance ML algorithms, see Fig. \ref{fig:ProbManipulation}. Let $K $ and $K_{0}$ denote two Kraus operators such that $K $ describes a desired transformation of quantum data and $K_0  = (\mathbb{I} - K ^\dag K )^{1/2}$ otherwise. Then, quantum data previously prepared by a quantum embedding $\phi$ denoted by $D_{\phi} = \{  (\rho_{\phi} (x_m), y_m)  \} $ can be transformed to $\widetilde{D}_{\phi} = \{  (\widetilde{\rho}_{\phi} (x_m), y_m) \} $ such that
\begin{equation}
\begin{aligned}
K: {\rho}_{\phi} (x_m) \mapsto   \widetilde{\rho}_{\phi} (x_m), ~ \mathrm{where} \\
\widetilde{\rho}_{\phi} (x_m)  = \frac{  K \rho_{\phi} (x_m) K^{\dagger}}{p(x_m)}, ~\mathrm{with} \\
p_s(x_m) = \mathrm{tr}[ K^{\dagger}K \rho_{\phi} (x_m) ]. \label{eq:fq}
\end{aligned}
\end{equation}
From Proposition, we can safely restrict the consideration to the case $\dim ( {\mathcal{H}}) = \dim ( \widetilde{\mathcal{H}})$ due to no advantage of utilizing a larger Hilbert space of quantum states. It is straightforward to see that, once the transformation is successful, a Kraus operator also leads to a kernel: 
\begin{equation}
\kappa(  \rho_{\phi}(x_n) , \rho_{\phi} (x_m) ) = \mathrm{tr} [ \widetilde{\rho}_{\phi}(x_n) \widetilde{\rho}_{\phi}(x_{m}) ]
\end{equation}
Therefore, the task to enhance ML algorithms for quantum data is to identify a Kraus operator $ K $ for manipulating quantum sample data beyond unitary transformations. If no enhancement occurs, one returns a trivial choice $K = \mathbb{I}$.

One may assert the weakness that a feature map in \eqref{eq:fq} is probabilistic so that for a large number of data points, the probability of its realization $ p_s(x_1)\times p_s(x_2)\times \cdots \times p_s(x_M)$ quickly falls to zero. A prescription is to construct a collective interaction over blocks of quantum data, a Kraus operator for quantum states $\rho_\phi(x_m)$.
\begin{equation}
K': \bigotimes_{m=1}^M {\rho}_{\phi} (x_m) \mapsto \bigotimes_{m=1}^M \widetilde{\rho}_{\phi} (x_m)
\end{equation}
However, the map above can be implemented with ancillary qubits that amount to given quantum data; a large set of resources is additionally required. Instead, we consider a Kraus operator for an ensemble of quantum data collectively,
\begin{equation}
K'': \sum_{m=1}^M {\rho}_{\phi} (x_m) \mapsto \sum_{m=1}^M  \frac{p_s(x_m)}{p_{\mathrm{succ}}} \cdot \widetilde{\rho}_{\phi} (x_m), \label{eq:KraForm}
\end{equation}
where $p_{\mathrm{succ}} = p_s(x_1, \dots, x_M)$ is a success probability that a Kraus operator is realized. The advantage is to reduce resources while maintaining a high enough probability. A consequence that is, in fact, not desired is to have new weights $p_s(x_m)/p_{\mathrm{succ}}$ according to success probabilities. We reiterate that the price to pay for a non-vanishing probability and resource advantage is to build a Kraus operator $K''$ for interaction among all data points.

\begin{figure*}[t]
	\begin{center}
		\includegraphics[angle=0, width=1. \textwidth]{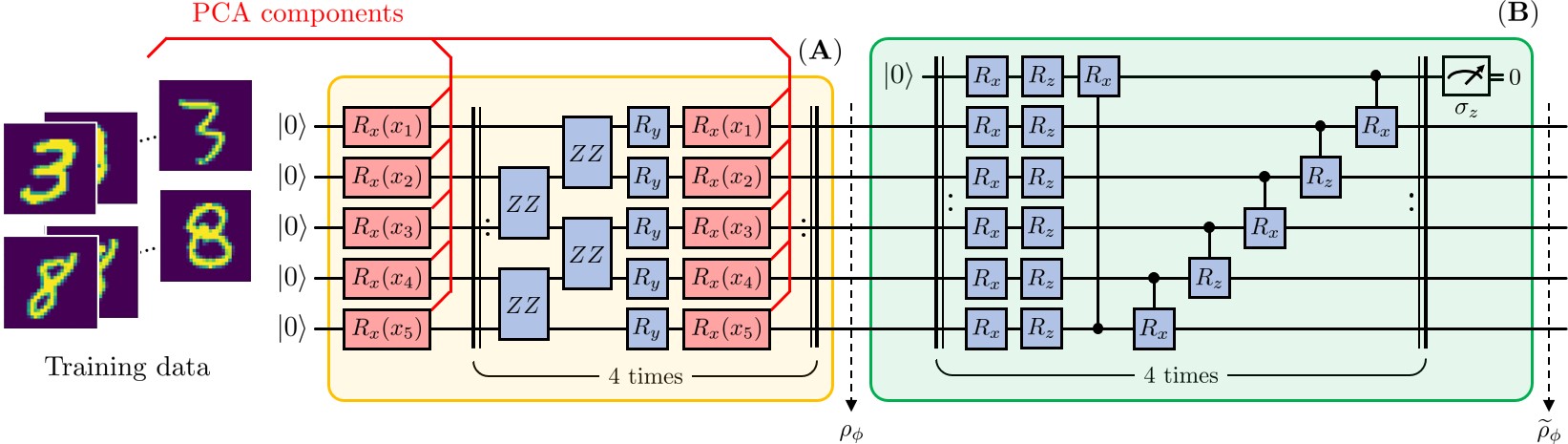}
		\caption{ (A) Principal component analysis (PCA) components $(x_1, \cdots, x_5)$ are extracted from handwritten images $"3"$ and $"8"$ in the MNIST dataset. A state $\rho_\phi$ is prepared after single-qubit rotations $R_x$ with PCA components followed by parameterized gates $R_y$ and $ZZ$. Parameterized gates are blue colored. (B) A feature map for quantum data realizes a non-unitary transform producing a state $\widetilde{\rho}_\phi$. The transformation is probabilistic since it is realized only when a single-qubit ancilla gives an outcome $0$. The feature map is implemented by parameterized single-qubit rotation $R_x$ and $R_z$ as well as controlled-$R_x$ gates. All these steps of quantum data manipulation are applied to the fashion-MNIST and cifar10 datasets in different PCA components. } \label{fig:SimQuaCircuit}
	\end{center}	
\end{figure*}

Then, by a Kraus operator in (\ref{eq:KraForm}), ensembles given in the beginning in \eqref{eq:OriState} are transformed as follwos,
\begin{equation}
\begin{aligned}
\widetilde{\rho} &= \sum_{m: y_m=1} \frac{p_s(x_m)}{p_s(y_m=1)} \widetilde{\rho}_{\phi}(x_m), \\
\widetilde{\sigma} &= \sum_{m: y_m=-1} \frac{p_s(x_m)}{p_s(y_m=-1)} \widetilde{\rho}_{\phi}(x_m), \label{eq:NewState}
\end{aligned}
\end{equation}
where for $i\in\{-1, 1\}$ 
\begin{equation}
p_s(y_m=i) = \sum_{m: y_m=i}p_s(x_m). \nonumber
\end{equation}
From the representer theorem in \eqref{eq:RepTheorem}, we obtain a fidelity classifier 
\begin{equation}
\widetilde{f}_{\mathrm{fid}}(x) = \mathrm{tr}[(\widetilde{\rho}-\widetilde{\sigma})\widetilde{\rho}_{\phi}(x)] \label{eq:NewFidClassifier}
\end{equation}
for which note that have set parameters in \eqref{eq:RepTheorem} as follows,
\begin{equation}
\begin{aligned}
\alpha_m=
\begin{cases} 
p_s(x_m)/p_s(y_m=1) ~ &\mathrm{if} ~ y_m=1 \nonumber\\
p_s(x_m)/p_s(y_m=-1) ~ &\mathrm{if} ~ y_m=-1. \nonumber
\end{cases}
\end{aligned}
\end{equation}
From (\ref{eq:9}), a weighted linear function is given by,
\begin{equation}
\begin{aligned}
L_{\widetilde{w}}(f(x_m), y_m) &= -\widetilde{w}_m f(x_m) y_m, ~\mathrm{where} \\
\widetilde{w}_m &= \frac{Mp_s(x_m)}{p_s(y_m=i)} ~ \mathrm{if} ~ y_m=i.
\end{aligned}
\end{equation}

Therefore, the Hilbert-Schmidt distance for two states $\widetilde{\rho}$ and $\widetilde{\sigma}$ is related to the empirical risk,
\begin{equation}
\hat{R}_{L_{\widetilde{w}}, \mathcal{D}}(\widetilde{f}_{\mathrm{fid}}) = -D_{\mathrm{hs}}(\widetilde{\rho}, \widetilde{\sigma}).  \label{eq:NewEmpRisk}
\end{equation}
It is shown that minimizing an empirical risk is equivalent to maximizing a Hilbert-Schmidt distance of two states.

\section{Demonstration}

\begin{figure*}[t]
	\begin{center}
		\includegraphics[angle=0, width=.78 \textwidth]{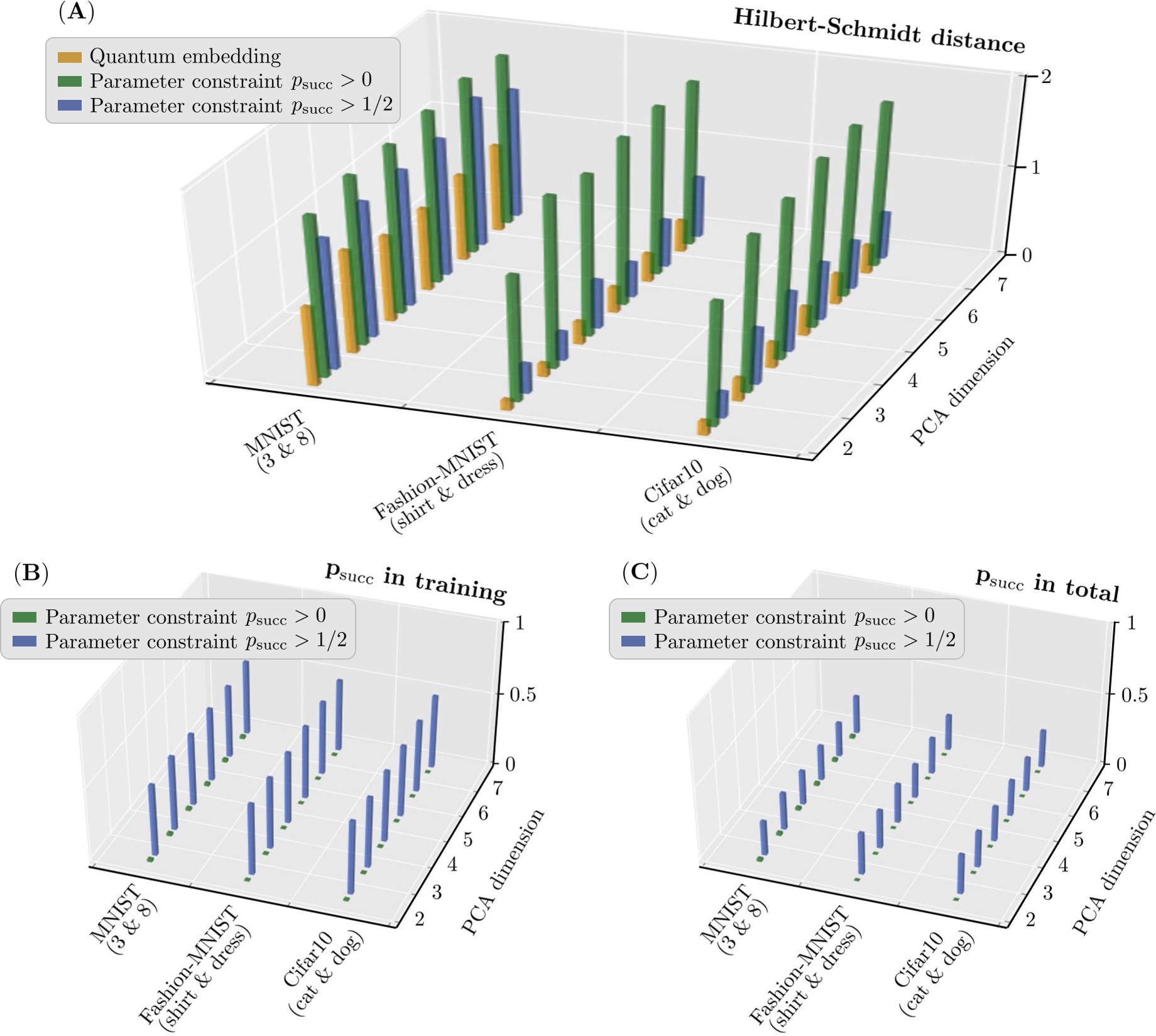}
		\caption{ (A) The enhancement of the classification is demonstrated for various PCA dimensions and datasets: MNIST, fashion-MNIST, and Cifar10. The parameter $c$ is from the constraint of a cut-off probability in \eqref{eq:CutEmpRisk}. We set the parameters $c=0$ and $c=1/2$ for the constraints $p_{\mathrm{succ}}>0$ and $p_{\mathrm{succ}}>1/2$, respectively. The proposed feature map for quantum data (green, blue) shows an improvement in the Hilbert-Schmidt distance over quantum embedding (orange) in all datasets. A constraint on the parameter $c$ in \eqref{eq:CutEmpRisk} can maintain a success probability higher and, at the same time, can lower an improved Hilbert-Schmidt distance. (B) A success probability in training can be significantly low (green). A constraint that $p_{\mathrm{succ}}>1/2$ by setting the cut-off parameter $c=1/2$ boosts the success probability higher than $1/2$ (blue). (C) Total success probabilities, that is, successful quantum feature maps in both training and test quantum data, are shown for all datasets. The constraint with the cut-off parameter can maintain a success probability higher.} \label{fig:SimResult}
	\end{center}	
\end{figure*}

In this section, we demonstrate the quantum binary classification with a feature map for quantum data, see Fig \ref{fig:OverScheQuanBinaClassification}. A quantum state can be prepared by a quantum embedding strategy that maps classical data $x$ to a quantum state $|\psi_x\rangle$ by a unitary transformation. An amplitude encoding is an instance as follows,
\begin{equation}
\rho_{\phi}(x_m) = |\psi_{x_m}\rangle\langle \psi_{x_m}|, ~\mathrm{where}~ |\psi_{x_m}\rangle = \sum_{i=0}^{N-1}x_{m,i}|i\rangle
\end{equation}
To apply a SWAP-test-based classifier, we write by $W$ a quantum circuit for the state preparation \cite{Schuld_2017},
\begin{equation}
W |0\rangle^{\otimes l} = \frac{1}{\sqrt{M}}\sum_{m=1}^{M} |m\rangle_L |\psi_{x_m}\rangle_T |\psi_{x} \rangle_t |0\rangle_S |s_m\rangle_C, \label{eq:InitState}
\end{equation}
where $l =(2+ \log M+ 2\log N)$ is the total number of qubits, see also Fig. \ref{fig:QuaCircuit}. The first register (L) labels input data, the second one (T) collects training quantum data, the third one (t) contains a test state, the fourth one (S) is an ancilla needed for a SWAP test, and the last one (C) denotes a binary classification of training data: $|s_m\rangle = |0\rangle$ for $y_m=1$ and $|s_m\rangle = |1\rangle$ for $y_m=-1$. Quantum registers (S) and (C) contain single qubits.

A feature map for quantum data, devised in \eqref{eq:KraForm}, applies to quantum data in \eqref{eq:InitState} to enhance the quantum binary classification. The map is facilitated by a unitary transformation $V(\theta)$ for each training and test data as well as an ancilla $F$.
\begin{equation}
\begin{aligned}
&V (\theta) \otimes V (\theta) \sum_{m=1}^{M}   |\psi_{x_m}\rangle_T |0\rangle_{F_T} |\psi_{x} \rangle_t |0\rangle_{F_t}  \nonumber \\
&= \sum_{m=1}^{M}  K(\theta) |\psi_{x_m}\rangle_T |0\rangle_{F_T}  K(\theta) |\psi_{x} \rangle_t |0\rangle_{F_t} \nonumber \\
&+ \sum_{m=1}^{M}  K(\theta) |\psi_{x_m}\rangle_T |0\rangle_{F_T}  K_0(\theta) |\psi_{x} \rangle_t |1\rangle_{F_t} \nonumber \\
& + \sum_{m=1}^{M}  K_0(\theta) |\psi_{x_m}\rangle_T |1\rangle_{F_T}  K(\theta) |\psi_{x} \rangle_t |0\rangle_{F_t} \nonumber \\
&+ \sum_{m=1}^{M}  K_0(\theta) |\psi_{x_m}\rangle_T |1\rangle_{F_T}  K_0(\theta) |\psi_{x} \rangle_t |1\rangle_{F_t} \nonumber
\end{aligned}
\end{equation}
An interaction $V(\theta)$ is to be trained to minimize the empirical risk in \eqref{eq:NewEmpRisk}. Once a desired one $K(\theta)$ is realized with an ancilla state $|0\rangle_F$, ensembles of training data are transformed to  $\widetilde{\rho}$ and $\widetilde{\sigma}$, respectively. Resulting quantum data with $|1\rangle_F$ that is not of interest will be discarded. In practice, an efficient unitary $V(\theta)$ may be designed as a parameterized quantum circuit and trained to minimize the cost function 
\begin{eqnarray}
{C}(\theta)={\hat{R}}_{L_{\widetilde{w}}, \mathcal{D}}({\widetilde{f}})_{\mathrm{}}.\label{eq:cost} 
\end{eqnarray}
If a feature map by a Kraus operator $K(\theta)$ cannot make any advantage, an optimal parameter after training would return $V(\theta) =\mathbb{I}$.

Then, a Hadamard classifier applies a SWAP test for states in registers $T$ and $t$. Qubits in registers in $S$ and $C$ are measured in the computational basis $\{| 0\rangle_z, |1\rangle_z \}$, respectively. The classifier in \eqref{eq:NewFidClassifier} can be constructed from outcomes in the registers. When an outcome $0$ $(1)$ is obtained in register $C$, states in registers $T$ are given by $\widetilde{\rho}$ $(\widetilde{\sigma})$ in \eqref{eq:NewState}. Given an outcome $0$ in register $C$, a measurement is performed in register $S$: probabilities of $0$ and $1$ are given by $(1\pm \langle \widetilde{\psi}_x | \widetilde{\rho} |\widetilde{\psi}_x\rangle )/2$, respectively. When an outcome in register $C$ is $1$, probabilities of $0$ and $1$ in register $S$ are given by $(1\pm \langle \widetilde{\psi}_x | \widetilde{\sigma} | \widetilde{\psi}_x\rangle )/2$. Hence, from the probabilities, one can find the classification function in \eqref{eq:NewFidClassifier} as  
\begin{equation}
\widetilde{f}_{\mathrm{fid}}(x) =  \langle \widetilde{\psi}_x | \widetilde{\rho} - \widetilde{\sigma} | \widetilde{\psi}_x \rangle.  \label{eq:qf}
\end{equation}

The quantum circuit that minimizes the empirical risk, or equivalently the cost function, shares a similar structure in their circuit constructions with the fidelity classifier. A distinction is that the gate $\widetilde{W}$ prepares a state for the training data as follows,
\begin{equation}
\widetilde{W}|0\rangle^{\otimes l'} = \frac{1}{\sqrt{M}}\sum_{m=1}^{M} |m\rangle_L |\psi_{x_m}\rangle_T |s_m \rangle_C, \label{eq:EmpInitState}
\end{equation}
where $l' = (1 + \log M + \log N)$ is the total number of qubits, see also Fig. \ref{fig:QuaCircuit}. The SWAP test between registers $T1$ and $T2$ results in the fidelity of states $\widetilde{\rho}$ and $\widetilde{\sigma}$, which are obtained according to an outcome in the register $C$. That is, a measurement of an observable $\sigma_z \otimes \sigma_z$ on the registers $C1$ and $C2$ finds the Hilbert-Schmidt distance of states $\widetilde{\rho}$ and $\widetilde{\sigma}$. Iterations of the steps optimize parameters in the Kraus operator so that the cost function converges to a local minimum value.

We consider the Iris and the modified National Institute of Standard and Technology (MNIST) dataset to demonstrate the advantage of a feature map for quantum data. For the Iris dataset, the Iris versicolor and Iris setosa are set as class A and B, respectively. Then, sepal lengths and sepal widths are given as input data. We select the $36$-th data in the Iris versicolor and the $34$-th data in the Iris setosa as the training dataset. We set the $29$-th data in the Iris versicolor for the test one. We have data in the following after normalization,
\begin{equation}
\begin{aligned}
\mathrm{Training~Data}~ : ~ \{((0.796, 0.607), 1), ((0,1), -1)\} \\
\mathrm{Test~Data}~ : ~ ((-0.557, 0.83), -1)
\end{aligned}
\end{equation}
We apply the amplitude encoding in the state preparation,
\begin{equation}
|\psi_{x_m}\rangle = R_y[2\cos^{-1}(x_{m,0})]|0\rangle. \nonumber
\end{equation}
In the case of the MNIST dataset, we select the images of $"3"$ and $"6"$ for binary classification. We rescale each image from $28 \times 28$ pixels to $4 \times 4$ pixels \cite{farhi2018classification}. Then, a $16$-dimensional state is prepared by amplitude encoding. In both cases, a feature map is trained by optimizing the empirical risk in \eqref{eq:NewEmpRisk}. The demonstration is summarized in Table \ref{table:Sim1}. A particular choice  $V(\theta)=\mathbb{I}$ corresponds to quantum embedding itself. We obtain success probabilities as $0.309$ and $0.331$ in for the IRIS and MNIST datasets, respectively.

\setlength{\tabcolsep}{3pt}
\renewcommand{\arraystretch}{1.2}
\begin{table}[b]
\caption{ Demonstration result of the feature map for quantum data. }
\begin{center}
\begin{tabular}{|c|c|c|c|c|}
\hline
\textbf{}&\multicolumn{2}{|c|}{\textbf{Iris Dataset}}&\multicolumn{2}{|c|}{\textbf{MNIST Dataset}} \\
\cline{2-5} 
\textbf{} & $V(\theta)=\mathbb{I}$ & $V(\theta) \neq \mathbb{I}$ & $V(\theta)=\mathbb{I}$ & $V(\theta) \neq \mathbb{I}$ \\
\hline
Empirical Risk & $-1.307$ & $-2.756$ & $-1.615$ & $-2.113$  \\
\hline
Classifier & $-0.718$ & $-0.987$ & $-0.113$ & $-0.414$ \\
\hline
Success Probability & $-$ & $0.309$ & $-$ & $0.331$  \\
\hline
Decision & $-1$ & $-1$ & $-1$ & $-1$  \\
\hline
\end{tabular}
\label{table:Sim1}
\end{center}
\end{table}

From an instance considered above, it turns out that the success probability of filtering quantum data decreases while the empirical risk improves. We suggest that one way to maintain a sufficiently high success probability is to leave it as a constraint in the empirical risk as follows,
\begin{equation}
\hat{R}'_{L_{\widetilde{w}}, \mathcal{D}}(f_{\mathrm{fid}}) = -\mathcal{D}_{\mathrm{hs}}(\widetilde{\rho}, \widetilde{\sigma}) + \lambda \max(0, c-p_{\mathrm{succ}}), \label{eq:CutEmpRisk}
\end{equation}
where $\lambda$ is a positive hyperparameter and $c\in[0, 1]$ is a cutoff probability. 

We further consider real-world image datasets with the constraint of a cutoff probability. We optimize both the feature map for quantum data and quantum embedding as in \cite{lloyd2020quantum}, see Fig. \ref{fig:SimQuaCircuit}. The empirical risk and success probabilities are then compared in three instances: quantum embedding, the feature map for quantum data, and a non-zero cutoff probability. We select three datasets for demonstrating the advantage of a feature map: i) handwritten images of $"3"$ and $"8"$ from the MNIST dataset, ii) images of shirts and dresses from the fashion-MNIST dataset, and iii) images of cats and dogs from the Cifar10 dataset. In the demonstration, $20$ data for each label are used for training, and the first $d$ PCA components are trained for quantum embedding. The results in Fig. \ref{fig:SimResult} show that a feature map for quantum data is more efficient than quantum embedding itself. Consequently, quantum data becomes more distinguishable in the classifier. Although a feature map for quantum data is genuinely probabilistic, adding the constraint of a cutoff probability to empirical risk can lead to a high enough success probability.

\section{Conclusion}
Finally, we reiterate that a {\it feature map for quantum data} presents a strategy for enhancing an ML algorithm for quantum states prepared according to a quantum feature map. It could be compared with the  {\it quantum embedding} of classical data shown in \cite{lloyd2020quantum}, which formulates training the state preparation to improve ML algorithms with quantum states. It is worth mentioning that, with some probability, our results always improve a quantum feature map or the quantum embedding, since a unitary transformation is an instance of a Kraus operator. Further enhancements of an ML algorithm may be envisaged by combining both strategies as shown in the numerical demonstration. Hence, a feature map for quantum data is a versatile tool to enhance existing ML algorithms. 

In conclusion, we have established a general framework for manipulating quantum data and then shown its applications to an ML algorithm for binary classification. We have also developed a circuit construction for a feature map for quantum data. In particular, our results shed light on the opportunity of restructuring quantum data on a Hilbert space beyond the constraint of contractivity. With some probability, quantum states can be distributed in a quantum feature space such that they are better distinguishable. We have applied the results to supervised learning for binary classification and demonstrated enhancements in the empirical risk. Our results present a versatile tool readily applicable to improve existing quantum ML algorithms.

\section*{Acknowledgment}
This work is supported by the National Research Foundation of Korea (Grant No. NRF-2021R1A2C2006309, NRF-2022M1A3C2069728, RS-2024-00408613, RS-2023-00257994) and the Institute for Information \& Communication Technology Promotion (IITP) (RS-2023-00229524).


\end{document}